\documentclass[11pt]{article} \usepackage{amsfonts}
\usepackage{times}
\usepackage{amssymb} \usepackage{amstext}
\usepackage{amsmath} \usepackage{theorem} \usepackage{algorithm}
\usepackage{graphicx}
\usepackage[left=1in,right=1in,top=1in,bottom=1in]{geometry}

 \newenvironment{proof}{{\bf
Proof:}}{\hfill\rule{2mm}{2mm}}

\newtheorem{definition}{Definition}[section]

\def\cA{{\cal A}}   
   
   \def\G{\Gamma}

\def\cE{{\cal E}} \def\cO{{\cal O}}
\def\cA{{\cal A}} \def\cF{{\cal F}} \def\cQ{{\cal Q}}

 \def\cF{{\cal F}}

\def\cD{{\cal D}} \newcommand{\ignore}[1]{}

\newcommand{\hide}[1]{{}}

\newtheorem{theorem}{Theorem} 

\newtheorem{lemma}[theorem]{Lemma}
\newtheorem{observation}[theorem]{Observation}

\newtheorem{remark}[theorem]{Remark}

  \newcommand{\beq}[1]{\begin{equation}\label{#1}}
\newcommand{\eeq}{\end{equation}}
\newcommand{\blem}[1]{\begin{lemma}\label{#1}}
\newcommand{\elem}{\end{lemma}}
\newcommand{\bth}[1]{\begin{theorem}\label{#1}}
\newcommand{\enth}{\end{theorem}}
\newcommand{\brem}[1]{\begin{remark}\label{#1}}
\newcommand{\erem}{\end{remark}}

\title{Exact counting of Euler Tours for Graphs of Bounded Tree-width}
\author{Prasad Chebolu\\ Department of Computer Science\\ University of 
Liverpool\\ \texttt{prasadc@liverpool.ac.uk} \and 
Mary Cryan\footnote{Supported by EPSRC grant EP/D043905/1}
  \footnote{Corresponding author}\\
  School of Informatics \\
  University of Edinburgh \\
  \texttt{mcryan@inf.ed.ac.uk}  \and Russell Martin\\ Department of Computer Science\\ University of Liverpool\\ \texttt{Russell.Martin@liverpool.ac.uk} }
\date{}

\begin{document}

\maketitle
\thispagestyle{empty}

\begin{abstract}
In this paper we give a simple polynomial-time algorithm to 
exactly count the number of Euler Tours (ETs) of any Eulerian 
graph of bounded treewidth.  The problems of counting ETs  are 
known to be $\sharp P$-complete for general graphs (Brightwell 
and Winkler, 2005~\cite{brightwell}).  To date, no polynomial-time 
algorithm for counting Euler tours of any class of graphs is known
except for the very special case of series-parallel graphs (which 
have treewidth 2).
\end{abstract}

\section{Introduction}

\subsection{Background}
Let $G=(V,E)$ be an undirected graph.  An {\em Euler tour (ET)} 
of~$G$ is any ordering $e_{\pi(1)}, \ldots, e_{\pi(|E|)}$ of the 
edges of~$E$, and of the directions of those edges, such that 
for every $1\leq i < |E|$, the target vertex of arc~$e_{\pi(i)}$ is 
the source vertex of~$e_{\pi(i+1)}$, and such that the target 
vertex of~$e_{\pi(|E|)}$ is the source of~$e_{\pi(1)}$.
We use $ET(G)$ to denote the set of Euler tours of $G$, where 
two Euler tours are considered to be equivalent if one is a cyclic
permutation of the other. It is a well-known fact that a given
undirected graph~$G$ has at least one Euler tour if and only 
if~$G$ is connected and every vertex~$v\in V$ has even degree.  

The problem of counting the number of Euler tours is considerably
more difficult than that of finding a single one, and few positive 
results exist.  Brightwell and Winkler~\cite{brightwell} have shown 
that the problem of counting ETs of undirected graphs is 
$\sharp P$-complete.  This is in contrast to the problem of
counting Euler tours of a directed Eulerian graph, where the number 
of Euler tours can be counted exactly in polynomial-time using the 
Matrix-Tree theorem\cite{borchardt} and the so-called ``BEST" Theorem 
(after de {\bf B}ruijn, van {\bf A}ardenne,  {\bf S}mith and {\bf T}utte, 
\cite{aardenne, smith}, though apparently the first two deserve credit 
as the original discoverers).   Creed\cite{creed} showed that the problem
on undirected graphs remains $\sharp P$-complete even for planar
graphs.  Therefore it is unlikely that a polynomial-time algorithm 
exists for the general class of undirected graphs or planar 
undirected graphs.    

It is well-known that many computational problems become tractable 
on graphs of bounded treewidth.  For example,
Noble~\cite{noble} showed that any point of the Tutte polynomial
can be evaluated on graphs of bounded treewidth in polynomial
time.  The number of ETs is not a point of the Tutte polynomial, nor 
do we know any reduction from Euler tours to a point of the Tutte 
(which would preserve treewidth).    However, Cryan et al.\cite{serpar} 
previously gave a direct polynomial-time algorithm for the special case 
of series-parallel graphs (which have treewidth 2). 
In this paper we extend this work to give a polynomial-time 
algorithm for graphs of treewidth~$d$, for bounded 
$d\in {\mathbb N}_0$.  

This is a preliminary paper; we believe that we can extend our 
results to compute any point of the Martin polynomial (whose
coefficients equal the number of vertex-pairings  of the graph
inducing~$\kappa$ components, for $\kappa =1, 2, \ldots, $).  
Makowsky and Mari\~no\cite{makowsky} conjectured in 2001 
that the Martin polynomial could be computed for graphs of 
bounded treewidth; however, no proof of this result exists yet. 

\subsection{Basic definitions}

Throughout this paper that we are working with a connected 
{\em Eulerian} graph $G=(V,E)$ - that is, a connected graph for 
which every vertex has even degree. We first introduce the well-known 
concepts of {\em tree decomposition} and {\em treewidth} of a graph:

\begin{definition}
A {\em tree decomposition} of a graph $G=(V,E)$ is a pair
$(\{X_i \mid i\in I\}, T=(I,F))$, where $T$ is a tree on~$I$
and where $X_i \subseteq V$ for all $i\in I$, such that the
following properties are satisfied:
\begin{enumerate}
\item $\cup_{i\in I} X_i = V$.
\item For every $(u,v)\in E$, there is some $i\in I$
such that $u,v \in X_i$.
\item For every $v\in V$, the set of nodes 
$\{i\in I \mid v\in X_i\}$ containing~$v$ induces a 
subtree of $T$. 
\end{enumerate}
The {\em width} of a tree decomposition $(\{X_i \mid i\in I\}, T=(I,F))$
is $\max_{i\in I} |X_i| - 1$.  The {\em treewidth} of a graph~$G$ is 
the minimum width over all tree decompositions of~$G$.
\end{definition}

The problem of computing the treewidth of of graph is well-known
to be fixed parameter tractable in the treewidth of the graph itself
(see for example, Bodlander\cite{bodlaender}).   
Moreover, many graph problems which are NP-complete in general 
can be solved in polynomial-time on graphs of bounded treewidth -
that is, graphs where the treewidth is at most~$k$ for
some fixed~$k$. There are also encouraging results 
for counting/sampling for $\sharp$P-complete problems.  
In particular, Noble~\cite{noble} has shown that any point
of the Tutte polynomial can be evaluated in polynomial time, 
for graphs of bounded treewidth.

A refinement of tree decomposition is the concept of a {\em nice
tree decomposition}: 
\begin{definition}\label{defn:nice}
A {\em nice tree decomposition} of a graph~$G$ is defined to be 
any tree decomposition $(\{X_i \mid i\in I\}, T=(I,F))$, where a 
particular node~$\rho$ of~$T$ is considered to be the root 
of~$T$, and each node~$i\in I$ is one of the four following types:
\begin{itemize}
\item[-] {\em leaf:} node~$i$ is a leaf of the tree~$T$, and
$|X(i)|=1$.
\item[-] {\em join:} node~$i$ has exactly two children~$j$
and~$j'$, and $X(i) = X(j) =  X(j')$. 
\item[-] {\em introduce:} node~$i$ has exactly one child~$i'$, 
and there is one vertex $v\in V$ such that 
$X(i) =X(i') \cup \{v\}$.
\item[-] {\em forget:} node~$i$ has exactly one child~$i'$,
and there is one vertex~$v\in V$ such that 
$X(i) =X(i') \setminus \{v\}$.
\end{itemize}
\end{definition}
It is well-known that if a graph~$G$ has a tree 
decomposition of width~$k$, then it also has a nice tree 
decomposition of width~$k$ and with~$O(n)$ nodes.  Also, for
any constant~$k$, there are well-known algorithms to find a nice 
tree decomposition of width~$k$ in polynomial-time (see, for example,
\cite{kloks}), for any given graph which has treewidth at most~$k$.

We will give an exact algorithm for counting ETs of a 
graph of treewidth~$k$, for any constant~$k\in {\mathbb N}_0$.

First observe there is a many-to-one relationship between 
ETs and {\em Eulerian Orientations} of a graph.  Let $EO(G)$ denote
the set of all Eulerian Orientations of the graph~$G$.
\begin{observation}
Let $G=(V, E)$ be any Eulerian multigraph. Then any Euler 
tour $T\in ET(G)$ induces a unique Eulerian Orientation 
on~$G$.
\end{observation}
The elements of~$ET(G)$ can be partitioned according to the 
particular Eulerian orientation they induce:
\begin{eqnarray}
ET(G) & =  & \cup_{\cE\in EO(G)} ET(\cE),\label{eq:EOdir} 
\end{eqnarray}
where~$ET(\cE)$ is the set of Euler tours of the directed 
graph given by~$\cE$.   We can refine this relationship 
further using the well-known relationship between Euler 
tours and in-Arborescences on directed Eulerian graphs:
\begin{theorem}[\cite{aardenne, smith}]
Let $\vec{G}=(V, \vec{E})$ be an Eulerian directed graph where each 
vertex~$v\in V$ has outdegree (and indegree) $d_v$.  Then for
any $r\in V$,
\begin{eqnarray}
|ET(\vec{G})| =  \left[\prod_{v\in V}(d_v-1)!\right]
|ARB(\vec{G},r)|.\label{eq:EOarb}
\end{eqnarray}
\end{theorem} 

Equations~\ref{eq:EOdir} and~\ref{eq:EOarb} motivate the 
following definitions of an {\em Orb} of an undirected connected 
Eulerian graph: 
\begin{definition}
Let $G=(V,E)$ be an connected Eulerian multigraph.  An {\em Orb} $\cO$ 
is a pair $\cO = (\cE, \cA)$, where $\cE\in EO(G)$ and  $\cA$ is
an in-directed Arborescence (rooted at some arbitrary $r\in V$) 
on~$\cE$. 
\begin{itemize}
\item Let $ORB(G)$ denote the set of all orbs~$(\cE, \cA)$ on~$G$;
\item For any $r\in V$, let $ORB(G,r)$ denote the set of all 
orbs~$(\cE, \cA)$ on~$G$, where~$\cA$ is rooted at~$r$.
\end{itemize}
For any $\cE\in EO(G)$, we identify~$\cE$ with the directed graph 
it induces on~$V$. 
\begin{itemize} 
\item Let~$ARB(\cE)$ denote the set of all in-Arborescences 
on~$\cE$;  
\item Let~$ARB(\cE, r)$ denote the set of all in-Arborescences 
on~$\cE$ which are rooted into~$r$.
\end{itemize}
\end{definition}

Combining~(\ref{eq:EOdir}) and~(\ref{eq:EOarb}), we have the 
following interpretation of the Euler tours of a given
undirected graph~$G$ with degrees $2d_v$:
\begin{eqnarray}
|ET(G)| & = & \left[\prod_{v\in V}(d_v-1)!\right]|ORB(G,r)|.
\end{eqnarray}
Therefore in order to count Euler tours of an undirected
Eulerian multigraph, it suffices to count Orbs of that
graph.   This is the approach we will take in this paper.

\section{Counting Orbs}

Consider a given Eulerian (multi)graph $G=(V, E)$ of 
treewidth~$k$, and let $(\{X(i) \mid i\in I\}, T=(I,F))$ 
be a nice tree decomposition of~$G$ with width~$k$.  
Let~$\rho\in I$ be the root of the tree decomposition,
and designate any vertex in~$X_\rho$ to be the 
distinguished vertex~$r$.   We will show how to evaluate
the cardinality of~$ORB(G,r)$ in~$O(n^{O(k)})$ time.

In the following two pages we introduce various concepts
which we will use in our dynamic programming algorithm, and 
make some simple observations.
\begin{definition}\label{def:Vi}
Let $(\{X(i) \mid i\in I\}, T=(I,F))$ be any tree decomposition of
the graph~$G=(V,E)$, and let $i\in I$.  We define the set~$V(i)$ to 
be~$\cup_{j\in T_i}X(j)$, the set of all vertices of~$G$ which appear in 
the subtree of the tree decomposition rooted at~$i$.
\end{definition}

\begin{definition}\label{def:EiEell}
Let $(\{X(i) \mid i\in I\}, T=(I,F))$ be any tree decomposition of
the graph~$G=(V,E)$.  For any $i\in I$, we define the following 
two subsets of~$E$:
\begin{eqnarray*}
E(i) & = & \{e=(u,v)\in E: u,v\in X(i)\}.  \\
E_\ell(i) & = & \{e=(u,v)\in E: u,v\in V(i), |\{u,v\}\cap X(i)|\leq 1\}.
\end{eqnarray*}
Let $G(i)$ and~$G_\ell(i)$ be the subgraphs of~$G$ induced by 
the edge sets above in turn.   Observe that for any~$i\in I$, 
the edge sets of~$G(i)$ and~$G_\ell(i)$ are edge-disjoint.
\end{definition} 

\begin{observation} 
Let $(\{X(i) \mid i\in I\}, T=(I,F))$ be any tree decomposition of
the graph~$G=(V,E)$ with root node~$\rho$.

Observe the union of~$G(\rho)$ and~$G_\ell(\rho)$ is the original 
graph~$G$.
\end{observation}

\begin{definition}
Let $(\{X(i) \mid i\in I\}, T=(I,F))$ be any tree decomposition of
the graph~$G=(V,E)$.  
For any $i\in I$ we define the two following disjoint subsets 
of~$X(i)$:
\begin{eqnarray*}
L(i) & = & \{x\in X(i): \exists \{v: (x,v)\in E_\ell(i)\}\};\\
U(i) & = & \{x\in X(i): \not\exists \{v: (x,v)\in E_\ell(i)\}\}.
\end{eqnarray*}
\end{definition}
For a given Orb, we can define induced orientations and arc 
sets at any node of the tree decomposition:
\begin{definition}\label{def:induce}
Let $G=(V,E)$ be a Eulerian multigraph with tree decomposition
$(\{X(i) \mid i\in I\}, T=(I,F))$.  Let~$\rho$ be the root of 
the decomposition, and consider the designated root $r\in X(\rho)$.  
Then for any orb $\cO=(\cE,\cA)\in ORB(G,r)$, and any node~$i$ of 
the decomposition, we define the following:
\begin{itemize}
\item[(i)] $\cE_\ell(i)$ is the restriction of~$\cE$ to the 
subgraph~$G_\ell(i)$;
\item[(ii)] $\cE(i)$ is the restriction of~$\cE$ to the 
subgraph~$G(i)$;
\item[(iii)] $\cA_\ell(i)$ is the restriction of~$\cA$ to
arcs in~$\cE_\ell(i)$.  
\item[(iv)] $\cA(i)$ is the restriction of~$\cA$ to arcs
belonging to~$\cE(i)$;
\end{itemize}
\end{definition} 

We now make a simple observation concerning the structures 
induced at the root node~$\rho$:
\begin{observation}\label{obs:forest}
Suppose we are given an Eulerian multigraph~$G=(V,E)$, with 
the nice tree decomposition $(\{X(i) \mid i\in I\}, T=(I,F))$,
having root~$\rho$.  Let $r\in X(\rho)$, and let 
$(\cE,\cA) \in ORB(G,r)$. Then 
\begin{itemize}
\item[(i)] $\cE_\ell(\rho)$ is an orientation on~$G_\ell(\rho)$ 
that satisfies the Eulerian condition at all $v\in V(\rho)\setminus 
X(\rho)$, but not necessarily at $v\in X(\rho)$.
\item[(ii)] ${\cE}(\rho)$ is an orientation on~$G(\rho)$, such that
for every $v\in X(\rho)$, the difference~$(out_{\cE(\rho)}(v) - 
in_{\cE(\rho)}(v))$ is equal to $(in_{\cE_\ell(\rho)}(v) 
-out_{\cE_\ell(\rho)}(v))$.
\item[(iii)] $\cA_\ell(\rho)$ is an in-directed forest on~$\cE_\ell(\rho)$ 
with some root set~$R\subseteq L(\rho)$, such that every 
$v\in (V(\rho)\setminus X(\rho)) \cup (L(\rho)\setminus R)$ 
has a out-arc in~$\cA_\ell(\rho)$.
\item[(iv)] $\cA(\rho)$ is a set of arcs from~$\cE(\rho)$, 
containing exactly one out-arc for every~$x\in ((X(\rho)\setminus 
L(\rho)) \cup R)\setminus \{r\})$ such that 
$\cA(\rho)$ is an in-directed tree rooted
at~$r$. 
\end{itemize}
\end{observation}
We will exploit Observation~\ref{obs:forest} in the design of
our dynamic programming algorithm to count Orbs of all Eulerian 
graphs with constant treewidth.  Our algorithm will count 
pairs of the form~$(\cE_{\ell}(\rho), \cA_\ell(\rho))$.  In fact,
we will count pairs of this type for all nodes~$i$ of our 
treewidth decomposition, partitioning the set of these pairs 
according to two parameters which we will call the 
{\em charge vector} and the {\em root vector}.

We now define the sets of Orientations we will 
consider in building the dynamic programming table. 
\begin{definition}\label{def:partialEOs}
Let $G= (V,E)$ be an Eulerian multigraph with tree decomposition 
$(\{X(i) \mid i\in I\}, T=(I,F))$, and consider any $i\in I$.
\begin{itemize}
\item Let~$D_\ell(i)$ denote the set of all orientations
of the edges of~$G_\ell(i)$ which are Eulerian at every
vertex $v\in V(i)\setminus X(i)$, but not necessarily at~$v\in X(i)$.
\item Let~$D(i)$ denote the set of all orientations (not
necessarily Eulerian) of the edges of $G(i)$.
\end{itemize}
\end{definition}
Note that parts~(i) and~(ii) of Observation~\ref{obs:forest} 
could be re-stated by saying that~$\cE_\ell(\rho)\in D_\ell(\rho)$
and~$\cE(\rho)\in D(\rho)$ respectively. 

We will partition the orientations of~$D(i)$ and~$D_\ell(i)$ 
in terms of the ``charge'' (outdegree - indegree) induced on 
the~$X(i)$  vertices by the orientation.
\begin{definition}\label{def:charge}
Let $G= (V,E)$ be a connected Eulerian multigraph with tree 
decomposition $(\{X_i \mid i\in I\}, T=(I,F))$, and let $i\in I$.  We 
define the following sets of ``charge vectors'':
\begin{itemize}
\item $C(i)\subseteq {\mathbb Z}^{|X(i)|}$ is the set of all
vectors~$c$ which can be generated by some orientation in~$D(i)$.   
\item $C_\ell(i)\subseteq {\mathbb Z}^{|X(i)|}$ is the set of all
  vectors~$c$ which can be generated by some orientation 
  in~$D_\ell(i)$.
\end{itemize}
\end{definition}
When discussing a specific orientation $\cD \in D(i)$ or
$\cD\in D_\ell(i)$, we write $c({\cD})$ for the charge vector 
induced by~${\cal D}$.
Observe that the vertex set induced by an orientation~$\cD\in 
D_\ell(i)$ is~$(V(i)\setminus X(i))\cup L(i)$, because the vertices 
of~$U(i)$ ($= X(i)\setminus L(i)$) are not endpoints of any edge 
in~$G_\ell(i)$.  
Therefore, in relation to Definition~\ref{def:partialEOs}, the only 
vertices for which~$\cD$ will violate the Eulerian property are 
the vertices in~$L(i)$, rather than all of~$X(i)$.  To be consistent 
with the charge vectors of~$C(i)$, we describe the charge 
vectors of~$C_\ell(i)$ in terms of~$X(i)$, although every
$c\in C_\ell(i)$ is guaranteed to have $c_v=0$ for all 
vertices~$v\in U(i)$. 

In counting $(\cE_\ell(i), \cA_\ell(i))$~pairs, we will 
partition the set of such pairs according to the charge 
vector of the orientation, but will also consider the 
root-status of the~$X(i)$ vertices in relation to~$\cA_\ell$.
For every $x\in X(i)$, there are three possible scenarios in 
regard to~$\cA_\ell(i)$:
\begin{itemize}
\item[(i)] $x\in U(i)$, in which case~$x$ 
does not belong to the induced multigraph~$G_\ell(i)$.  
\item[(ii)] $x\in L(i)$, and~$x$ is a (possibly isolated) 
root of a subtree in~$\cA_\ell(i)$;
\item[(iii)]  $x\in L(i)$, and~$x$ is not the root of a subtree 
in~$\cA_\ell(i)$.  There is some $y\in L(i)$ such that~$y$
is a root in~$\cA_\ell(i)$, such that~$x$ belongs to~$y$'s 
subtree (via an out-arc~$x\rightarrow z$ for some 
$z\in V(i)\setminus X(i)$).
\end{itemize} 
We use this distinction between vertices of~$X(i)$ to define
the concept of a {\em root vector} for an arbitrary node~$i$ 
of the tree decomposition:
\begin{definition}\label{def:rootvector} 
Let $G=(V,E)$ be a connected Eulerian multigraph, and suppose 
we have a tree decomposition $(\{X(i) \mid i\in I\}, T=(I,F))$ of~$G$
of width~$k$.  Let~$i\in I$.  We define the set~$S(i)$ of 
{\em root vectors} induced by $(\cE_\ell(i), \cA_\ell(i))$~pairs
to be the set~$S(i)$ of all $s\in |X(i)|^{|X(i)|}$ 
satisfying the following properties:
\begin{itemize}
\item For every $x\in U(i)$, $s_x=x$;
\item There is at least one~$x\in L(i)$ such that $s_x = x$;
\item Let $R(s)\subseteq L(i)$ be the set of vertices~$x\in R(s)$ 
such that $s_x= x$. Then for every $y\in L(i)\setminus R(s)$, 
$s_y\in R(s)$. 
\end{itemize}
\end{definition}
By Observation~\ref{obs:forest}, and by our discussion above
Definition~\ref{def:rootvector}, every induced forest~$\cA_\ell(i)$
of an orb~$(\cE, \cA)\in ORB(G,r)$ is consistent with a unique 
$s\in S(i)$.  For a specific forest~$\cF$ on~$G_\ell(i)$, we 
will write $s(\cF)$ to denote the vector of~$|X(i)|^{|X(i)|}$ 
which indicates, for each of the vertices in~$X(i)$, the root 
of the tree in~$\cF$ which contains it.  Note that the information
carried by the vectors of~$S(i)$ could also be encoded as vectors
in~$|L(i)|^{|L(i)|}$, however to have consistency with the 
charge vectors, we assume the root vectors have length~$|X(i)|$,
while enforcing the constraint that all vertices~$x$ of~$U(i)$ 
have~$s_x$ set to~$x$.

Now we define the concepts of {\em forests} and {\em forest Orbs} 
for nodes of the tree decomposition.
\begin{definition}\label{def:partialforests}
Let $G=(V,E)$ be an Eulerian multigraph, and suppose we have 
a tree decomposition $(\{X(i) \mid i\in I\}, T=(I,F))$ of~$G$
of width~$k$.  Let~$i\in I$.  

Let $\cD\in D_\ell(i)$.  We define a {\em forest} with respect 
to~$\cD$ to be any in-directed forest~$\cF$ on~$\cD$ such that 
\begin{itemize}
\item There is some set~$R(\cF)\subseteq L(i)$ such that~$R(\cF)$ 
is the set of roots of~$\cF$;
\item Every~$v\in (V(i)\setminus X(i))\cup (L(i) \setminus R(\cF))$ 
has an out-arc in~$\cF$.
\end{itemize}

We write~$FOR(\cD)$ to denote the set of all forests on~$\cD$ 
(wrt~$i$). 

We define a {\em forest Orb} to be any pair~$(\cD, \cF)$
such that $\cD\in D_\ell(i)$ and $\cF \in FOR(\cD)$. 
\end{definition}

Our algorithm will construct a table $\Psi(i)$, indexed 
by pairs~$(c,s)$  for $c\in C(i)$ and $s\in S(i)$. 
The table will store the value 
\begin{eqnarray*}
\psi(i, c, s) & = & \sum_{\cD\in D_\ell(i)} \sum_{\cF\in FOR(\cD)} 
I_{c(\cD)=c, s(\cF)=s}.
\end{eqnarray*}
In the following section, we will show how to build the table~$\Psi(i)$ 
for every node~$i$ of the tree decomposition in polynomial-time, 
in a bottom-up fashion of the tree~$(I,F)$.  However we first 
show how, once we have this table constructed for the root~$\rho$ 
of the tree decomposition, we can then compute the number of Orbs
of the original graph~$G$.  
\bigskip

{\bf NOT DONE YET}.  It will be similar to the proof of the 'forget' 
case.

We now have the following observation about ``charge 
vectors''. 
\begin{observation}\label{obs:vectors1}
Suppose $G=(V,E)$ is a graph (or multi-graph) with tree 
decomposition $(\{X_i \mid i\in I\}, T=(I,F))$, and let $i\in I$. 
Then every~$c\in C(i)$ satisfies $c(x)\in \{-d_{G_\ell(i)}(x),
-d_{G_\ell(i)}(x)+2, \ldots, d_{G_\ell(i)}(x)-2, d_{G_\ell(i)}(x)\}$.
\end{observation}

Combining Observation~\ref{obs:vectors1} together with bounded
treewidth, we can derive specific bounds on the size of~$|C(i)|$.
\begin{observation}\label{obs:vectors2}
Suppose $G=(V,E)$ is a {\em simple} graph with tree 
decomposition $(\{X_i \mid i\in I\}, T=(I,F))$ of treewidth~$k$, 
and let $i\in I$. Then by simplicity, we know $d_{G_\ell(i)}(x)
\leq (n-1)$ for all $x\in X_i$.  Therefore by 
Observation~\ref{obs:vectors1}, $|C(i)|\leq n^k$.
\end{observation}

For the case of multi-graphs, we have a lesser observation:
\begin{observation}\label{obs:vectors3}
Suppose $G=(V,E)$ is a multi-graph with tree decomposition 
$(\{X_i \mid i\in I\}, T=(I,F))$ of treewidth~$k$, and let 
$i\in I$. Let $m=|E|$. Then we have $d_{G_\ell(i)}(x) \leq (m-1)$ 
for every $x\in X_i$.  Hence $|C(i)|\leq m^k$.
\end{observation}

Finally, we present the following bound on the number of
root vectors:
\begin{observation}\label{obs:rootvectors}
Suppose $G=(V,E)$ is a multi-graph with tree decomposition 
$(\{X_i \mid i\in I\}, T=(I,F))$ of treewidth~$k$, and let 
$i\in I$.  Then the number of root vectors~$|S(i)|$ satisfies 
the bound $S(i) \leq |L(i)|^{|L(i)|} \leq k^k$.
\end{observation}

\subsection{Our algorithm}\label{ssec:EOalg}

We now discuss the bottom-up computation of the table $\Psi(i)$,
storing the values $\psi(i, c,s)$ for all~$c\in C(i)$, $s\in S(i)$.

Note that if~$G = (V,E)$ is simple, then by Observations~\ref{obs:vectors2}
and~\ref{obs:rootvectors}, the table~$\Psi(i)$ contains at 
most~$n^kk^k$ entries, where $n=|V|$.  Alternatively, if~$G$ is not 
necessarily simple, then by Observations~\ref{obs:vectors3} 
and~\ref{obs:rootvectors}, $\Psi(i)$ contains at most~$m^kk^k$
entries.

We now show how to build~$\Psi(i)$ for all nodes of the tree 
decomposition $(\{X_i \mid i\in I\}, T=(I,F))$.  This is done in 
a bottom-up dynamic programming fashion, with the tables for 
node~$i$ only being built {\em after} the corresponding tables 
for the child node (or nodes) of~$i$ have already been 
constructed.  Recall that every node of a nice treewidth 
decomposition has at most two child nodes.  

\subsubsection{Leaf}

\noindent
In the case of a leaf node~$l$, we have $X(l)=\{w\}$ for some 
vertex~$w\in V$.  $G_\ell(l)$ is an empty graph with no vertices 
or edges.  There is exactly one charge vector in~$C(l)$ - this is 
the vector~$c^*$ of length~$|X(l)|=1$ which assigns charge-0 to~$w$.   

To consider possible sets of root-vectors, note 
that~$X(l)=\{w\}$, and~$L(l) = \emptyset$. Therefore the only
root vector in~$S(l)$ is the vector~$s^*$ of length~1 which 
assigns $s^*_w=w$.

Finally, the only orientation on~$G_\ell(l)$ to satisfy~$c^*$ (or 
indeed any charge vector) is the empty one~$\cD^*$; also the only 
forest on~$\cD^*$ to satisfy~$s^*$ is again the empty forest~$\cF^*$
consisting of no arcs, and the single isolated vertex~$w$.  Hence
the table~$\Psi(l)$ consists of the following single entry:
$$\psi(l,c^*,s^*)=1.$$

\subsubsection{Introduce}

For the case of introduce, our current node~$i\in I$ has a 
single child~$i'$, and $X(i) = X(i')\cup \{w\}$ for some
$w \not \in X(i')$.  By the properties of a nice treewidth 
decomposition, we know that for every $v\in V(i') \setminus X(i')$, 
there is no edge of the form $(w,v)$ in~$G$.  Therefore the adjacent 
vertices to~$w$ are all either in~$X(i)$ or in $V\setminus V(i)$,
and $L(i)=L(i')$.  The graph~$G_\ell(i)$ is identical to 
$G_\ell(i')$.   If we adopt the convention that the entry for~$w$ 
is at the end of the charge vectors in~$C(i)$, then 
$$C(i) = \{c.0: c\in C(i')\}.$$  
Now we consider the set of root vectors~$S(i)$ in relation 
to~$S(i')$.  Assuming that the entry for~$w$ will be stored 
at the end of the root vectors for~$i$, then by 
Definition~\ref{def:rootvector} and by~$L(i)=L(i')$, we 
have $$S(i)=\{s.w: s\in S(i')\}.$$

Next we consider the value of~$\psi(c.0, s.w)$, for
any $c\in C(i'), s\in S(i')$ in relation to the table~$\Psi(i')$
which has previously been computed.  Given that $G_\ell(i) =
G_\ell(i')$, and by Definition~\ref{def:partialEOs}, we know 
that an orientation~$\cD\in D_\ell(i)$ satisfies $c(\cD)=c.0$ 
if and only if we have $c(\cD)=c$ at node~$i'$.   Also, for
any $\cF\in FOR(\cD)$, we have $s(\cF) =s.w$ at~$i$ if and
only if we have~$s(\cD)=s$ at node~$i'$.  Hence the values for
the table $\Psi(i)$ are, for every $c\in C(i')$, every~$s\in S(i')$,
$$\psi(i, c.0, s.w) = \psi(i', c, s).$$
  
\subsubsection{Forget}

For the case of forget, our current node~$i\in I$ has a 
single child~$i'$, and $X(i) = X(i')\setminus \{w\}$ for some
$w \not \in X(i')$.   Note that the charge vectors in~$C(i)$
will be of length~$1$ less than those in~$C(i')$, because 
any $c\in C(i)$ will not include an entry for~$w$.  Similarly
the root vectors of~$S(i)$ will be of length 1-less than 
those in~$S(i')$ for the same reason.   

In using the values of table $\Psi(i')$ to create the table~$\Psi(i)$,
we will need a few more definitions.  First of all, we define some
subclasses of edges:
\begin{itemize}
\item  For any~$u, v\in V$, we define $E_{u,v} = \{e\in E, e=(u,v)\}$,
and $m_{u,v}=|E_{u,v}|$;  
\item $E_{\ell}(i',w) = \cup_{v\in V(i')\setminus X(i')} E_{w,v}$;
\item $E(i',w)= \cup_{x\in X(i')\setminus \{w\}} E_{w,x} = 
\cup_{x\in X(i)}E_{w,x}$.
\end{itemize} 
Observe that $E(i',w)$ and~$E_\ell(i', w)$ are mutually
disjoint and their union is the set of all edges adjacent to~$w$.

We will use the definitions above to relate forest-Orbs for~$i$ with 
forest-Orbs for~$i'$.  We represent a forest-Orb of~$FOR(i)$ as 
$(\cD, \cF)$ for $\cD\in D_\ell(i)$ and $\cF\in FOR(\cD)$.  For 
every such~$(\cD, \cF)$, we define
\begin{itemize}
\item $\cD'$ to be the restriction of~$\cD$ to the graph~$G_\ell(i')$
(ie, to the edges of~$E_\ell(i')$);
\item $Q_{w,v}\subseteq E_{w,v}$ to be the edges of~$E_{w,x}$ directed 
{\em away} from~$w$ in~$\cD$, for any~$v$ adjacent to~$w$ in~$G$.  
We also define $q_{w,v}=|Q_{w,v}|$;
\item $\cF'$ to be the restriction of~$\cF$ to the arcs~$\cD'$;
\end{itemize}    

The following theorem specifies the relationship between elements 
of~$FOR(i)$ and~$FOR(i')$ in the forget case:
\begin{theorem}\label{thm:forget}
Let $G=(V,E)$ be an Eulerian multigraph with tree decomposition
$(\{X(i) \mid i\in I\}, T=(I,F))$, and let $i\in I$ be a forget
node such that $X(i)=X(i') \setminus \{w\}$ for some $w\in X(i')$,
where~$i'$ is the single child of~$i$.  

Suppose~$\cD$ is an orientation of the edges of~$G_\ell(i)$ 
and~$\cF$ is some set of arcs of~$\cD$.  Then $(\cD, \cF)\in 
FOR(i)$ with charge vector~$c\in C(i)$ and root vector~$s\in S(i)$ 
{\em if and only if}
$\cD$ is the disjoint union of~$\cD'\in D_\ell(i')$ 
and some orientation~$Q$ of~$E(i,w)$ (with induced values~$q_{w,x}$ 
for $x\in A_w\cap X(i)$), {\em and} $\cF$ is the 
disjoint union of some~$\cF'\in FOR(D')$ and some arc set 
$\cQ\subseteq Q$ such that all the following conditions hold:
\begin{itemize}  
\item[(a)] $(\cD', \cF')\in FOR(i')$;
\item[(b)] $\sum_{x\in X(i)}(m_{w,x}-2q_{w,x}) = c(\cD')_w$;
\item[(c)] $c_x = c(\cD')_x$ for~$x \in X(i)\setminus A_w$, 
$c_x = c(\cD')_x -2q_{w,x}+m_{w,x}$ for~$x\in X(i)\cap A_w$; 
\item[(d)] The forests~$\cF$, $\cF'$ and their root sets 
$R=\{x\in L(i): s(\cF)_x=x\}$ and $R' = \{x\in L(i'): 
s(\cF')_x=x\}$ are related in one of the following ways: 
\begin{itemize}
\item[$w\not\in R'$:] In this case~$R=(R'\setminus L)\cup U$ 
for some pair of sets~$U$ and~$L$ such that:
\begin{itemize}
\item[(i)] $L\subset (R'\cap \{x\in A_w: q_{w,x} < m_{w,x}\})
\setminus \{s(\cF')_w\}$; 
\item[(ii)] $U\subseteq U(i')\cap A_w$ such that $U \supseteq
U(i')\cap \{x\in A_w: q_{w,x} = m_{w,x}\}$.  
\end{itemize}     
The set~$\cQ$ is the union of exactly one arc of the form 
$(x\rightarrow w)$ (from the~$m_{w,x}-q_{w,x}$ possibilities), 
for every $x\in L \cup ((A_w\cap U(i'))\setminus U)$.
\item[$w\in R'$:] In this case $R= (R'\setminus (L\cup\{w\}))\cup 
U\cup (\{w^*\}\cap U(i'))$, for a vertex~$w^*$ and sets~$L,U$ 
such that:  
\begin{itemize} 
\item[(I)] $w^*\in (X(i)\cap \{x\in A_w: q_{w,x}>0\})\setminus 
\{x\in L(i'): s(\cF')_x=w\}$;
\item[(II)] $L\subset (R'\cap \{x\in A_w: q_{w,x} < m_{w,x}\})
\setminus \{w, w^*, s(\cF')_{w^*}\}$; 
\item[(III)] $U\subseteq U(i')\cap A_w$ such that $U \supseteq
U(i')\cap \{x\in A_w: q_{w,x} = m_{w,x}\}$.  
\end{itemize} 
The set~$\cQ$ is the union of exactly one arc of the form 
$(x\rightarrow w)$ (from the~$m_{w,x}-q_{w,x}$ possibilities), 
for every $x\in L \cup ((A_w\cap U(i'))\setminus (U\cup \{w^*\}))$, 
together with one arc of the form~$(w\rightarrow w^*)$.
\end{itemize}
\end{itemize}
\end{theorem} 
\begin{proof}
We prove this Lemma in two parts.

{\bf if:} We first show the ``if'' of our claim.  
Assume that we have $\cD'\in D_\ell(i')$, $\cF'\in  FOR(\cD')$, 
$Q\in D(i)$ and $\cQ$ such that conditions (a)-(d) are satisfied.
We will show that then $(\cD'\cup Q, \cF'\cup \cQ)$ is an element 
of~$FOR(i)$ with charge vector~$c$ and root vector~$s$, and with
root set~$R$ as described.

We proceed in two stages.  We first prove that~$\cD'\cup Q$ is
an orientation of~$D_\ell(i)$ with the claimed charge vector~$c$.
We know that~$\cD'$ is a orientation on~$G_\ell(i')$
which is Eulerian at every vertex $v\in V(i')\setminus X(i')$, 
and which has some charge~$c(\cD')_x$ for every $x\in L(i')$
(and 0 charge at every~$x\in U(i')$).  
By definition, the graph~$G_\ell(i)$ is equal to~$G_\ell(i')$
together with the set of edges~$E(i,w)= \{e=(w,x): x\in X(i)\}$.  
$Q$ is an orientation on the set~$E(i,w)$.  Therefore 
$\cD'\cup Q$ is an orientation on the graph~$G_\ell(i)$.

Consider the charge induced by $\cD'\cup Q$ on~$G_\ell(i)$,
by considering 4 cases: $V(i')\setminus X(i')$, $w$, $L(i)$
and~$U(i)$. 
\begin{itemize}
\item The charge induced on any~$v\in  V(i')\setminus X(i')$
is 0, because~$\cD'$ induces charge~$0$ on these vertices,
and these vertices do not appear in $E(i,w)$ (and hence~$Q$
induces no charge).
\item The charge induced on~$w$ is 
$c(\cD')_w + \sum_{x\in U(i')\cap A_w}(2q_{w,x}-m_{w,x}),$
where the~$\cD'$ contributes~$c(\cD')_w$, and the second 
expression is the contribution from~$Q$.  

Under assumption~(b), this evaluates to~$0$.
\item Let $x\in L(i) = L(i')\setminus \{w\}$. 
These vertices belong to~$\G_\ell(i)$, to~$\cD'$ and to~$Q$.  The
charge induced by~$\cD'\cup Q$ is $c(\cD')_x + (m_{w,x}-2q_{w,x})$
if $x\in A_w$, and $c(\cD')_x$ otherwise. In both these cases
assumption~(c) implies an overall charge of~$c_x$, as required.
\item The charge induced on any vertex of $U(i) = U(i')\setminus 
A_w$ is 0.  These vertices do not belong to~$\cD'$, $Q$ or 
to~$G_\ell(i)$.
\end{itemize}
Note that~$V(i)\setminus X(i)= (V(i')\setminus X(i'))\cup \{w\}$.  
Hence the overall orientation~$\cD'\cup Q$ is Eulerian at all 
vertices~$V(i)\setminus X(i)$.  The charge vector~$c(\cD'\cup Q)$
has the value~$c_x$ at all $x\in L(i)$, and~$0$ at all $x\in U(i)$,
as required.   Hence $\cD'\cup Q\in D_\ell(i)$ with charge
vector~$c$.

Now consider $\cF'\cup \cQ$, where~$\cF'\in FOR(\cD')$, $\cQ$ is
an subset of the arcs in~$Q$, and $\cF', \cQ$ satisfy (d).  We
will show that under these circumstances $\cF'\cup \cQ$ is a 
forest with the claimed root set~$R$ on~$\cD'\cup Q$.  To show 
that $\cF'\cup \cQ$ is a forest with root set~$R$, we must show:
\begin{itemize}
\item[($\alpha$)] That no vertex~$x\in R$ has an outgoing arc 
in~$\cF'\cup \cQ$;
\item[($\beta$)] That every vertex~$v\in V(i)\setminus (U(i)\cup R)$ has
exactly one outgoing arc in~$\cF'\cup \cQ$;   
\item[($\gamma$)] That $\cF'\cup \cQ$ contains no directed cycle. 
\end{itemize}

We will prove ($\alpha$)-($\gamma$) individually, first 
considering the $w\not\in R'$ case, then the $w\in R'$ case.
\bigskip

\noindent
($\alpha$): 
Our goal is to show (in both the $w\not\in R'$ case and the 
$w\in R'$ case) that no vertex of~$R$ has an outgoing arc 
in~$\cF'$ or in $\cQ$.  By definition, $\cF'$ contains outgoing
arcs for every $v\in V(i')\setminus (R'\cup U(i'))$.  There is no 
outgoing arc for any~$x\in R'\cup U(i')$ in~$\cF'$.  Note that 
regardless of whether $R = (R'\setminus L)\cup U$ ($w\not\in R'$ 
case) or $R= (R'\setminus (L\cup\{w\}))\cup U \cup (\{w^*\}\cap U(i'))$ 
($w\in R'$ case), we have $R\subseteq R'\cup U(i')$.  So there are 
no outgoing arcs for vertices of~$R$ in~$\cF'$.  We now consider the 
arcs of~$\cQ$.

In the $w\not \in R'$ case, we have
$R'=(R\setminus L)\cup U$ for~$L, U$ as specified in (d).  The
arcs of~$\cQ$ are of the form~$(x\rightarrow w)$ for $x\in 
L \cup ((U(i')\cap A_w)\setminus U)$.  Note that $(R'\setminus L)$ 
and~$U$ each have an empty intersection with $L \cup ((U(i')\cap A_w)
\setminus U)$.  Hence for $w\not \in R'$, no vertex of~$R$ has an 
outgoing arc in~$\cF'\cup \cQ$. 

In the $w\in R'$ case, $R$ is the disjoint union of 
$(R'\setminus (L\cup\{w\}))$, $U$ and $\{w^*\}\cap U(i')$. 
The set~$\cQ$ contains the arc~$(w\rightarrow w^*)$, together
with an arc~$(x\rightarrow w)$ for every~$x\in L \cup 
((U(i')\cap A_w)\setminus (U\cup\{w^*\}))$.   Now note that
$w\not \in R$, hence we need not consider the 
arc~$(w\rightarrow w^*)$ further.  Next note that $L$ has an 
empty intersection with each of $(R'\setminus (L\cup\{w\}))$, 
$U$ and~$\{w^*\}\cap U(i')$, and therefore $L\cap R = \emptyset$.
Finally, note that $(U(i')\cap A_w)\setminus (U\cup \{w^*\})$
also has an empty intersection with each of 
$(R'\setminus (L\cup\{w\}))$, $U$ and~$\{w^*\}\cap U(i')$.  
Therefore no arc of~$\cQ$ is outgoing from a vertex of~$R$.
Therefore in the case of $w\in R'$, no vertex of~$R$ has an 
outgoing arc in~$\cF'\cup \cQ$, and~($\alpha$) holds.   
\bigskip

\noindent
($\beta$): We must show that every vertex  
$v\in V(i)\setminus (U(i)\cup R)$ has exactly one outgoing arc
in~$\cF'\cup \cQ$.  We first note that that $V(i') = V(i)$.
For every $v\in V(i')\setminus X(i')$, we know that 
$v$ has exactly one outgoing arc in~$\cF'$.  Also, if 
$v\in V(i')\in X(i')$ there is no outgoing arc from~$v$ 
in~$\cQ$ (since $v\neq w$, and $v\not\in A_w\cap X(i')$).
Hence every $v\in V(i')\setminus X(i')$ has exactly one 
outgoing arc in $\cF'\cup \cQ$, as required.

We will now show that every $x\in X(i')\setminus 
(U(i)\cup R)$ has exactly one outgoing arc in~$\cF'\cup \cQ$.
First observe that $U(i) = U(i')\cap \overline{A_w}$. 

First consider the case~$w\not \in R'$.  In this case 
$R=(R'\setminus L)\cup U$. We now partition~$X(i')$ into six 
sets as follows:
\begin{center}
\begin{tabular}{llllll} 
$R'\cap L$, & $R'\setminus L$, & $L(i')\setminus R'$, &
$U(i')\cap \overline{A_w}$, & $U$, & $(U(i')\cap A_w)\setminus U$.
\end{tabular} 
\end{center} 
Then $X(i')\setminus (U(i)\cup R)$ is the union of the three
disjoint sets $(R'\cap L)$, $L(i')\setminus R'$ (which includes~$w$) 
and $(U(i')\cap A_w)\setminus U$.  We will show that every vertex 
in these sets has exactly one outgoing arc in~$\cF'\cup \cQ$.  For 
$x \in R'\cap L$, we know that~$x$ has no outgoing arc in~$\cF'$.
However by construction, $\cQ$ contains exactly one 
arc~$(x\rightarrow w)$.  For $x\in L(i')\setminus R'$, 
$\cF'$ contains an outgoing arc for~$x$; however, there is no arc 
leaving~$x$ in~$\cQ$, so again~$x$ has exactly one outgoing arc 
in~$\cF'\cup \cQ$.  Finally, for 
$x\in (U(i')\cap A_w)\setminus U$, $x$ is not in $G_\ell(i')$ and 
therefore has no outgoing arc in~$\cF'$; however, by construction,
$\cQ$ contains exactly one arc of the form~$(x\rightarrow w)$
leaving~$x$.  So in all three cases, there is one outgoing arc 
for~$x$ in~$\cF' \cup \cQ$, as required.  Hence ($\beta$) holds
in the $w\not \in R'$ case. 

Next consider the case~$w\in R'$.  In this case we have 
$R=(R'\setminus (L\cup\{w\}))\cup U\cup (\{w^*\}\cap U(i'))$.
We partition~$X(i')$ into eight sets in this case:
\begin{center}
\begin{tabular}{llllllll} 
$\{w\}$,\hspace{-0.06in} & $R'\cap L$,\hspace{-0.06in} & 
$R'\setminus (L\cup \{w\})$,\hspace{-0.06in} & 
$L(i')\setminus R'$,\hspace{-0.06in} & 
$U(i')\cap \overline{A_w}$,\hspace{-0.06in} & 
$U$,\hspace{-0.06in} & $\{w^*\}\cap U(i')$,\hspace{-0.06in} & 
$(U(i')\cap A_w)\setminus (U\cup \{w^*\})$.
\end{tabular} 
\end{center} 
Then by definition of~$R$, $x\in X(i')\setminus (R\cup U(i))$ if and 
only if~$x$ belongs to one of $\{w\}$, $R'\cap L$, $L(i')\setminus R'$
and $(U(i')\cap A_w)\setminus (U\cup \{w^*\})$.  We show that every
vertex in each of these four sets has exactly one outgoing arc.
For the vertex~$w$, $\cF'$ contains no outgoing arc from~$w$ 
(as $w\in R'$), but~$\cQ$ contains one arc of the 
form~$(w\rightarrow w^*)$, hence $\cF' \cup \cQ$ has exactly one 
outgoing arc from~$w$.  Let~$x\in R'\cap L$. In this case, 
by~$x\in R'$, we know that $x$ has no outgoing arc in~$R'$; also,
by~$x\in L$, we know that~$\cQ$ contains exactly one arc leaving~$x$
(an arc of the form~$(x\rightarrow w)$).  Hence for~$x\in R'\cap L$,
$\cF'\cup \cQ$ contains exactly one arc leaving~$x$.  Now suppose
$x\in L(i')\setminus R'$.  In this case~$\cF'$ already contained
one outgoing arc from~$x$.  However, by $x\in L(i')\setminus R'$
we know $x\not\in L$ and $x\not \in U(i')$, hence none of the arcs
of~$\cQ$ is outgoing from~$x$.  So $\cF'\cup \cQ$ contains exactly
one arc leaving~$x$ for~$x\in L(i')\setminus R'$.  Finally assume
$x\in (U(i')\cap A_w)\setminus (U\cup \{w^*\})$.  By $x\in U(i')$,
we know such an~$x$ will have no outgoing arc in~$\cF'$.  By
definition of~$cQ$, there is exactly one arc of the 
form~$(x\rightarrow w)$ in~$\cQ$ for such an~$x$.  So again, 
there is one outgoing arc in~$\cF'\cup \cQ$ for every~$x\in 
(U(i')\cap A_w)\setminus (U\cup \{w^*\})$.   Hence ($\beta$)
holds in the case of~$w\in R'$.
\bigskip

($\gamma$): Next we show that there is no simple directed cycle 
in~$\cF'\cup \cQ$ (together with~$\alpha$ and~$\beta$, this 
will imply the non-existence of any cycle in the undirected image
of~$\cF'\cup \cQ$).   By our assumption that~$\cF'$ is a forest  
on~$\cD'$, there can be no directed cycle in~$\cF'$.  Therefore
any simple directed cycle that could exist in~$\cF'\cup \cQ$ would
need to contain at least one arc from~$\cQ$. 

We will treat the $w\not \in R'$ and $w\in R'$ cases separately.
One observation which we will use repeatedly is the following -
if~$x\in L(i')\setminus R'$, there is exactly one~$z\in R'$ 
such that there is a path from~$x$ to~$z$ in~$\cF'$.  This is
because there existed $z, z', z\neq z'$ satisfying this condition,
this would imply a directed path in~$\cF'$ between~$z$ and~$z'$,
where both~$z$ and~$z'$ are in~$R'$ (and hence neither has a
outgoing arc in~$\cF'$).  Note the unique~$z$ is $z=s(\cF')_x$.

We consider the $w\not \in R'$ case first.  In this case, every 
arc of~$\cQ$ is of the form~$(x\rightarrow w)$ for some~$x\in X(i)$.
Therefore if a simple directed cycle exists in~$\cF'\cup \cQ$, then 
it must contain exactly one~$\cQ$ arc.  Also, since no vertices 
of~$U(i')$ belong to~$\cF'$, the arc must be~$(x^*\rightarrow w)$ for
some~$x^*\in L$.   Consider such a hypothetical cycle consisting 
of~$(x^* \rightarrow w)$, together with a directed path~$p$ from~$w$ 
to~$x^*$ lying entirely in~$\cF'$. Moreover, since $w\in L(i')
\setminus R'$ and $x^*\in L\subseteq R'$, by our observation above
we must have~$x^*=s(\cF')_w$.  Now observe that (d)(i) excludes~$s(\cF')_w$ 
from being a member of~$L$, therefore there can be no arc 
from~$x^*=s(\cF')_w$ to~$w$ in~$\cQ$.  This proves that for the 
$w\not \in R'$ case, there can be no simple directed cycle 
in~$\cF'\cup \cQ$.

We now consider the case of $w\in R'$.  In this case, $\cQ$ consists  
of one arc~$(w\rightarrow w^*)$, together with one arc of the 
form~$(x\rightarrow w)$ for every $x\in L\cup 
((U(i')\cap A_w)\setminus (U\cup \{w^*\}))$.  A simple directed cycle 
may visit~$w$ at most once, hence a simple directed cycle may either 
contain exactly one~$\cQ$ arc (either~$(w\rightarrow w^*)$ or one
of the~$(x\rightarrow w)$ arcs) or exactly two~$\cQ$ arcs, where
in the latter case this must be one of the~$(x\rightarrow w)$ arcs
followed immediately in the cycle by $(w\rightarrow w^*)$.  We 
consider each of these cases in turn.  First consider a hypothetical 
cycle consisting of the arc~$(w\rightarrow w^*)$ and a path~$p$ 
in~$\cF'$ from~$w^*$ to~$w$.  By existence of an outgoing path 
from~$w^*$ in~$\cF'$, we can deduce that $w^*\in L(i')\setminus R'$.
We know~$w\in R'$.  Then by our observation, we must 
have~$w=s(\cF')_{w^*}$.  Now recall that (d)(I) specifies that~$w^*$ 
cannot be any vertex which lies in the subtree of~$w$ in~$\cF'$.  
So we have a contradiction for the case of a cycle containing 
$(w\rightarrow w^*)$ and no other~$\cQ$ arcs. Next consider a 
hypothetical cycle consisting of one arc of the form~$(x\rightarrow 
w)$ from~$\cQ$ and a path in~$\cF'$ from~$w$ to~$x$.  Observe that 
the existence of a path leaving~$w$ in~$\cF'$ would imply that~$w$ 
must be an element of~$L(i')\setminus R'$, in direct contradiction 
to the fact that~$w\in R'$.  Hence there is no simple directed 
cycle in~$\cF'\cup \cQ$ containing exactly one~$\cQ$ arc.  Consider 
the final possibility for a simple cycle in $\cF'\cup \cQ$,
where we have $(x\rightarrow w)$ (from $\cQ$) for some~$x$
followed directly by the arc $(w\rightarrow w^*)$, and then by
a path~$p$ in~$\cF'$ from~$w^*$ to~$x$.  Note that for such a
path to exist in~$\cF'$, given that $x, w^*\in X(i')$, we must 
have~$w^*\in L(i')\setminus R'$ and $x \in R', x=s(\cF')_{w^*}$.
Now recall that by~$x\in R'$, we know $x\in L(i')$, and therefore
the arc~$(x\rightarrow w)$ of~$\cQ$ is from~$x\in L$.  However,
(d)(II) specifically states that~$s(\cF')_{w^*}$ is not
an element of~$L$.  Hence we have a contradiction.  So in all
three possible subcases of~$w\in R'$ we have shown that a cycle 
is impossible in~$\cF'\cup \cQ$. 

\bigskip

{\bf only if:}  It is also true that given an orientation~$\cD\in 
D_\ell(i)$ and a forest~$\cF\in FOR(\cD)$ with root set~$R$, 
charge vector~$c\in C(i)$ and root vector~$s\in S(i)$, that 
conditions (a)-(d) are satisfied.  Note this is the easier direction 
of the proof.
\end{proof}

We now apply Theorem~\ref{thm:forget} to the calculation 
of~$\psi(i,c,s)$ for $c\in C(i), s\in S(i)$ in the forget case.
We know that if~$(\cD,\cF)$ is in~$FOR(i)$ with charge vector~$c$
and root vector~$s$ if and only if all of conditions~(a)-(d) hold
for $(\cD', \cF') = (\cD(i'),\cF(i'))$ and $(Q, \cQ)= 
(\cD^{(i)}, \cF^{(i)})$.  We now make some observations 
concerning conditions~(a)-(d):

\begin{observation}
Let $G=(V,E)$ be an Eulerian multigraph with tree decomposition
$(\{X_i \mid i\in I\}, T=(I,F))$, and let $i\in I$ be a forget
node such that $X(i)=X(i') \setminus \{w\}$ for some $w\in X(i')$,
where~$i'$ is the single child of~$i$.  

Consider the task of counting pairs~$(\cD, \cF)\in FOR(i)$ with
charge vector~$c\in C(i)$ and root vector~$s\in S(i)$.  
For any $(\cD', \cF')\in FOR(i')$, every orientation~$Q$ on~$E(i',w)$, 
and every subset $\cQ \subseteq Q$ with the induced values~$q_{w,x},
x\in A_w\cap X(i)$, conditions (b)-(d) can be expressed solely in 
terms of $c$, $s$, $c'= c(\cD')$, $s'=s(\cF')$, the edge 
counts~$m_{w,x} = |E_{w,x}|$ for~$x\in X(i)$, the out-of-$w$
edge counts $q_{w,x}=|Q_{w,x}|$ for~$x\in X(i)$, and finally, the 
collection of arcs~$\cQ$.   
\end{observation}
\begin{proof} 
That this is true is immediately clear for conditions~(b) and~(c),
which are describe in terms of~$c$ and~$c'$.

Condition~(d) takes more consideration.  First observe that we 
can test whether~$w\in R'$ or $w\not\in R'$ (and identify which
set of tests need to be carried out) by checking whether
$s(\cF')_w$ is equal to~$w$ or not.  Also note that we already 
know the sets of vertices~$L(i')$, $U(i')$, $A_w$, $L(i)= (L(i')\setminus
 \{w\})\cup U(i')\cap A_w$, $U(i)=U(i')\cap \overline{A_w}$, and
$X(i) = L(i) \cup U(i)$, as these can be determined from the tree 
decomposition of~$G$. 

Suppose first that we are considering the case of $w\not\in R'$, 
hence we need to check (i), (ii) and also the details for~$\cQ$.
We have the two root vectors~$s\in S(i)$ and $s\in S(i')$; 
therefore from these root vectors we can identify 
$R = \{x\in L(i): s_x=x\}$ and $R' = \{x\in L(i'): s_x=x\}$.
Given the relationship that exists between~$R$ and~$R'$, we 
must have $U = R\setminus R'$ and $L= R'\setminus R$.  Now 
we can check that (i) holds for~$L$ in polnomial-time, by taking 
the intersection of~$R'$ and $\{x\in A_w: q_{w,x}>0\}$ and 
excluding~$s'_w$ from this set.  We can check that (ii) holds
by calculating the set~$U(i')\cap A_w$ and checking that~$U$ 
is contained in his set; then calculating the set $U(i')\cap 
\{x \in A_w: q_{w,x}=m_{w,x}\}$ and checking that this set is 
contained in~$U$.   Finally, if (i) and (ii) have been passed,
we check that~$\cQ$ is the union of a single arc~$(x\rightarrow w)$
for every~$x\in L \cup ((A_w\cap U(i'))\setminus U)$ by 
examining~$\cQ$ directly.

Next suppose we are considering the case of~$w\not\in R'$, so
need to check conditions (I)-(III) and also check which arcs 
lie in~$\cQ$.  First note again that we can calculate~$R'$ 
from~$s'$ and~$R$ from~$s$.  To check (I), we first identify the 
vertex~$w^*$ (this will be the target of the only arc outgoing 
from~$w$ in~$\cQ$).  We then compute the sets $X(i)\cap 
\{x\in A_w: q_{w,x}>0\}$ and $\{x\in L(i'):s'_x=w\}$.  Then we 
check that $w^*$ is in the first set, but not the second.  
Next we determine $L$ and $U$.  If we take $R\setminus R'$,
this evaluates to~$U\cup (\{w^*\}\cap  U(i'))$.  We already 
know the vertex~$w^*$, and whether it belongs to~$U(i')$ or 
not, therefore, we can recover the set~$U$ by deleting~$w^*$
is necessary.  If we take~$R\setminus R'$ this evaluates to
$L\cup \{w\}$.  Excluding~$w$ gives us~$L$.  To check 
condition (II), we calculate $R'\cap \{x\in A_w: q_{w,x}<m_{w,x}\}$
and exclude any of the vertices~$w,w^*, s'_{w^*}$ which appear 
in this set.  Then we check that every vertex of~$L$ appears 
in the computed set.   The test (III) is exactly the same as 
test~(ii) of the case $w\not\in R'$, and we evaluate it in exactly
the same way.  Finally we check that~$\cQ$ contains the necessary 
arcs by checking that it contains exactly the set of arcs described.
\end{proof}
\bigskip

We now discuss how to compute the table~$\Psi(i)$.

We start by initialising the value $\psi(i,c,s)$ to~0, for 
every $c\in C(i)$ and $s\in S(i)$.

Next we iterate through the table~$\Psi(i')$ one entry at a time,
using the value~$\psi(i', c',s')$ (in conjunction with all possible 
orientations~$Q$ of~$E(i,w)$, and all relevant sets~$\cQ$, to 
increase the value of~$\psi(i,c,s)$ for any values of~$c,s$ which 
satisfy Theorem~\ref{thm:forget} (in conjunction with $Q$, $\cQ$)
and in relation to~$c',s'$). For each $c',s'$, we perform the 
following steps. 
\begin{itemize}     
\item[(i)] We check whether~$s'_w$ is~$w$ (ie, whether $w\in R'$) or 
otherwise ($w\not\in R'$).
\item[(ii)] We consider each vector $q\in \prod_{x\in A_w\cap X(i)}
\{0,\ldots, m_{w,x}\}$ such that 
$$c'_w = \sum_{x\in A_w\cap X(i)} (m_{w,x}-2q_{w,x})$$
in turn, and compute the weight 
$$\chi(q) =_{def} \prod_{x\in A_w\cap X(i)}\binom{m_{w,x}}{q_{w,x}},$$
which is the number of different orientations of the edges
of~$E(i,w)$ which have exactly~$q_{w,x}$ of the $E_{w,x}$ edges 
oriented away from~$w$, 
For each vector~$q$, we define $c^*$, the charge vector of~$\cF'\cup Q$
for any orientation~$Q$ consistent with~$q$, to be 
$$c^*_x = \begin{cases} 
c'_x  & \mbox{if }x\in X(i)\cap \overline{A_w} \\
c'_x + (m_{w,x}-2q_{w,x}) & \mbox{if }x\in X(i)\cap A_w
\end{cases}.$$  
\item[(iii)]
We now have two cases, depending on whether~$w\in R'$ or not. 

$w\not\in R'$:
\begin{itemize}
\item[(a)] We first consider each set 
$L\subset (R'\cap \{x\in A_w: q_{w,x} < m_{w,x}\}) \setminus \{s'_w\}$
in turn; and also consider each subset $U$ such that 
$U\subseteq U(i')\cap A_w$ and $U \supseteq U(i')\cap 
\{x\in A_w: q_{w,x} = m_{w,x}\}$ in turn.  Note that by $|R'|\leq k$, 
$s'_w \in R'$ and $|U(i')|\leq (k-1)$, we know there are at most 
$2^{k-1}$ possible pairs of sets~$L,U$ to be considered, which
is constant (since the treewidth~$k$ is constant).

By Theorem~\ref{thm:forget}, recall that in order for
$\cF'\cup \cQ$ to be a forest on the orientation~$\cD\cup Q$
of~$G_\ell(i)$, that~$\cQ$ must be the union of exactly one 
arc~$(x\rightarrow w)$ for every~$x\in L \cup 
((U(i')\cap A_w)\setminus U)$.   The number of ways we can 
choose these arcs is
$$\kappa(L,U,q) = \prod_{x\in L}(m_{w,x}-q_{w,x})
\prod_{x\in (U(i')\cap A_w)\setminus U} (m_{w,x}-q_{w,x}).$$
Observe that by our conditions on~$L$ and~$U$, we know that for
every~$x\in L$ and every~$x\in (U(i')\cap A_w)\setminus U$, that
$m_{w,x}- q_{w,x}$, the number of arcs from~$x$ to~$w$, is 
non-zero.

\item[(b)] We now define the root vector of~$\cF'\cup \cQ$, for any~$\cQ$
which induces the vector~$q$.  This will be $s^*=s(q,L,U)$, defined
as
$$s^*_x = \begin{cases}
s'_w & \mbox{if }x\in L(i')\setminus \{w\}, s'_x\in L\\  
s'_x  & \mbox{if }x\in L(i')\setminus \{w\}, s'_x\not\in L \\
s'_w  & \mbox{if }x\in (U(i')\cap A_w)\setminus U\\
x & \mbox{if }x\in (U(i')\cap \overline{A_w})\cup U
\end{cases}.$$  

\item[(c)]
Finally, we add the value~$\chi(q)\times \kappa(L,U,q)\times 
\psi(i', c', s')$ to the table entry for~$\psi(i, c^*, s^*)$.
\end{itemize}

$w\in R'$:
\begin{itemize}
\item[(a)]
We consider every possible~$w^*$ in the set
$(X(i)\cap \{x\in A_w: q_{w,x}>0\})\setminus \{x\in L(i'): s(\cF')_x=w\}$
in turn. Note that for certain orientations~$Q$, the set of 
potential~$w^*$ vertices may be empty.  In these cases, we 
skip part~(iii) and try another~$q$ vector (as described in~(ii)).

Conditional on this~$w^*$, we consider every possible
$L\subset (R'\cap \{x\in A_w: q_{w,x} < m_{w,x}\}) \setminus \{w, w^*, 
s'_{w^*}\}$ in turn; and also consider each subset $U$ such that 
$U\subseteq U(i')\cap A_w$ and $U \supseteq U(i')\cap 
\{x\in A_w: q_{w,x} = m_{w,x}\}$ in turn.  By $|X(i)|\leq k$,
there are at most~$k-1$ possible values for~$w^*$.  For each 
particular~$w^*$, there are at most $2^{k-1}$ possible pairs of 
ets~$L,U$ to be considered, which is constant.  So we will consider
at most~$(k-1)2^{k-1}$ triples $(w^*, L, U)$.

By Theorem~\ref{thm:forget} in the $w\in R'$ case, recall that 
for $\cF'\cup \cQ$ to be a forest on the orientation~$\cD\cup Q$
of~$G_\ell(i)$, that~$\cQ$ must be the union of one arc of the 
form~$(w\rightarrow w^*)$, together with exactly one 
arc~$(x\rightarrow w)$ for every~$x\in L \cup 
((U(i')\cap A_w)\setminus (U\cup \{w^*\}))$.   The number of ways 
we can choose these arcs is
$$\kappa(w^*,L,U,q) = q_{w,w^*} \prod_{x\in L}(m_{w,x}-q_{w,x})
\prod_{x\in (U(i')\cap A_w)\setminus (U\cup \{w^*\})}(m_{w,x}-q_{w,x}).$$
Observe that by our conditions on~$L$ and~$U$, we know that for
every~$x\in L$ and every~$x\in (U(i')\cap A_w)\setminus 
(U\cup \{w^*\})$, that $m_{w,x}- q_{w,x}$, the number of arcs 
from~$x$ to~$w$, is non-zero.
\item[(b)]
Next we compute the root vector~$s^*$ of~$\cF'\cup \cQ$, for the 
current~$\cQ$:
$$s^*_x = \begin{cases} 
s'_{w^*}  & \mbox{if }x\in L(i')\setminus \{w\}, s_x' \in L\cup \{w\}\\
s'_x  & \mbox{if }x\in L(i')\setminus \{w\}, s_x'\not\in L\cup \{w\} \\
s'_{w^*} & \mbox{if }x\in (U(i')\cap {A_w})\setminus U\\
x & \mbox{if }x\in (U(i')\cap \overline{A_w})\cup U 
\end{cases}.$$ 
\item[(c)]
Finally, we add the value~$\chi(q)\times \kappa(w^*,L,U,q)\times 
\psi(i', c', s')$ to the table entry for~$\psi(i, c^*, s^*)$.
\end{itemize}
\end{itemize}

\subsubsection{Join}

In the case of a join node~$j$, we know that~$j$ has two child 
nodes~$i$ and~$i'$, and that $X(j) = X(i) = X(i')$.  Observe that 
$V(j) = V(i) \cup V(i')$.  Also note that by the rules
of a join for a nice tree decomposition, that $V(j) \setminus X(j)$ is
the disjoint union of $(V(i) \setminus X(i))$ and $(V(i')\setminus X(i'))$.
Also, $G$ does not contain any edges connecting vertices of 
$V(i)\setminus X(i)$ with vertices of $V(i')\setminus X(i')$. 
Therefore the graph $G_\ell(j)$ is the disjoint union of the graphs
$G_\ell(i)$ and $G_\ell(i')$.

Obseve now that the charge vectors in~$C(j)$ have the same length
and are indexed by the same set of vertices~$X(j)$, as the charge
vectors of~$C(i)$ and of~$C(i')$.   Also, the root vectors in~$S(j)$
have the same length and are indexed by the same set of vertices~$X(j)$, 
as the root vectors of~$S(i)$ and of~$S(i')$.  We now have the following 
observation about the decomposition of any forest Orb~$(\cD, \cF)\in FOR(j)$:
\begin{observation}\label{obs:join}
Let $j$ be a join-node of the nice tree decomposition
$(\{X(i) \mid i\in I\}, T=(I,F))$ of the Eulerian multi-graph~$G$,
and let~$i$ and~$i'$ be the child nodes of~$j$.  Then $(\cD,\cF)$
is a forest Orb of~$j$ with charge vector~$c$ and root vector~$s$
if and only if $\cD$ is the disjoint union of~$\widehat{\cD}\in D(i)$
and~${\cD}'\in D(i')$, and~$\cF$ is the disjoint union of
$\widehat{\cF}\in FOR(\widehat{\cD})$ and $\cF'\in FOR(\cD')$ such
that
\begin{itemize}
\item[(a)] $c_x = c(\widehat{\cD})_x + c({\cD}')_x$ for all~$x\in C(j)$;
\item[(b)] For every $L(i)\cap L(i')$, at least one of 
$s(\widehat{\cF})_x= x$ and $s(\cF')_x=x$ holds;
\item[(c)] For every $x \in L(i)\setminus \widehat{R}$, $y\in L(i')\setminus 
R'$ either $s(\widehat{\cF})_x\neq y$ or $s(\cF')_y\neq x$;   
\item[(d)] If we let~$\widehat{s}$ denote the root vector 
of~$\widehat{\cF}$ and~$\cF'$ denote the root vector of~$\cF'$,
then~$s$ satisfies the following: 
$$s_x = \begin{cases} 
\widehat{s}_x  & \mbox{if }x\in L(i), \widehat{s}_x\in L(i)\setminus L(i') \\
s'_{\widehat{s}_x} & \mbox{if }x\in L(i), \widehat{s}_x\in L(i)\cap L(i') \\
s'_x  & \mbox{if }x\in L(i'), s_x'\in L(i')\setminus L(i) \\
\widehat{s}_{s'_x} & \mbox{if }x\in L(i'), {s}'_x\in L(i)\cap L(i') \\
x & \mbox{if }x\in X(j) \setminus (L(i)\cup L(i')) 
\end{cases}.$$
\end{itemize}   
\end{observation}
\begin{proof}
Condition (a) is trivial.

For the forest conditions, it is not difficult to check that 
if~$\cF\in FOR(j)$, then all of (b), (c), (d) hold.  

To prove that (b), (c), (d) imply that~$\widehat{\cF}\cup \cF'$
is a forest (with the root vector~$s$) on~$\widehat{\cD}\cup \cD'$,
note that properties~($\alpha$) and~($\beta$) for a forest follow 
easily from the fact that~$\widehat{\cF}$ and~$\cF'$ are forests, 
and from~(b).  Checking~($\gamma$) takes a little bit more work, 
but is implied by~(c).
\end{proof}
\bigskip

We now describe how to fill table~$\Psi(j)$  when~$j$ is a 
join node.  First, for every $c\in C(j), s\in S(j)$, we 
initialise~$\psi(j,c,s)$ to~0.
Next we iterate through the table~$\Psi(i)$ one entry at a time,
using the value~$\psi(i, \widehat{c},\widehat{s})$ in conjunction
with the table~$\Psi(i')$ to increase the value of~$\psi(j,c,s)$
entries. For each $\widehat{c},\widehat{s}$, we perform the 
following steps. 
\begin{itemize}
\item[(i)] We compute the following sets using~$G$ and~$\widehat{s}$:\\
$L(i) =_{def} \{x\in X(i): A_x\cap (V(i)\setminus X(i))\neq \emptyset\}$,
$\widehat{R} = \{x\in L(i): \widehat{s}_x=x\}$.       
\item[(ii)] We consider each index~$(i', c', s')$ of~$\Psi(i')$ in 
turn.
\begin{itemize}
\item We define the charge vector $c^* =_{def} \widehat{c}+c'$.
\item We compute the sets $L(i') =_{def} \{x\in X(i'): A_x\cap 
(V(i')\setminus X(i'))\neq \emptyset\}$,
${R}' = \{x\in L(i'): {s}'_x=x\}$.    
\end{itemize}
\item[(iii)]
If properties (b) and (c) hold for~$\widehat{s}, s'$ then we
compute~$s^*$ as 
$$s^*_x = \begin{cases} 
\widehat{s}_x  & \mbox{if }x\in L(i), \widehat{s}_x\in L(i)\setminus L(i') \\
s'_{\widehat{s}_x} & \mbox{if }x\in L(i), \widehat{s}_x\in L(i)\cap L(i') \\
s'_x  & \mbox{if }x\in L(i'), s_x'\in L(i')\setminus L(i) \\
\widehat{s}_{s'_x} & \mbox{if }x\in L(i'), {s}'_x\in L(i)\cap L(i') \\
x & \mbox{if }x\in X(j) \setminus (L(i)\cup L(i')) 
\end{cases},$$
and add a value of~$1$ to the current value for~$\psi(j,c^*, s^*)$,
otherwise we do nothing.

Then we return to (ii) and consider a new index of~$\Psi(i')$.
\end{itemize}

Observe that for the join computation, the issue of whether~$G$
was a simple graph or a multi-graph is not relevant (except in 
bounding the size of the table).

\hide{
\section{Eulerian Orientations}

\subsection{Definitions}

Some definitions first.  

For any fixed node~$i$ of the treewidth decomposition, we
partition the edge set~$E$ of~$G$ as follows:
\begin{eqnarray*}
E(i) & = & \{e=(u,v)\in E: u,v\in X_i\}.  \\
E_\ell(i) & = & \{e=(u,v)\in E: u,v\in V_i, |\{u,v\}\cap X_i|\leq 1\}.\\
E_u(i) & = & \{e=(u,v)\in E: u,v\in V\setminus (V_i\setminus X_i), 
|\{u,v\}\cap X_i|\leq 1\}.
\end{eqnarray*}
Let $G(i), G_\ell(i)$ and $G_u(i)$ be the subgraphs of~$G$ 
induced by the edge sets above in turn.   Observe that for
any~$i\in I$, the edge sets of $G(i), G_\ell(i)$ and~$G_u(i)$ 
are edge-disjoint and the union of the three graphs is 
exactly~$G$.

Consider any Eulerian orientation ${\cal E}\in EO(G)$.
We can partition~${\cal E}$ into ${\cal E}_u(i), {\cal E}(i)$
and ${\cal E}_\ell(i)$, these being the orientations induced
by~$\cal E$ on the three subgraphs $G(i), G_\ell(i)$ and $G_u(i)$
respectively.  Observe that for ${\cal E}\in EO(G)$, that for
every~$i\in I$, both ${\cal E}_u(i)$ and ${\cal E}_\ell(i)$ 
satisfy the property of being Eulerian at every $v\not \in X_i$. 
We will see that this fact, together with constant treewidth~$k$,
will allow us to derive a polynomial-time algorithm for counting
EOs. 

\begin{definition}\label{def:partialEOs}
Let $G= (V,E)$ be a given graph with tree decomposition 
$(\{X_i \mid i\in I\}, T=(I,F))$, and consider any $i\in I$.
\begin{itemize}
\item Let~$O(i)$ denote the set of all orientations (not
necessarily Eulerian) of the edges of $G(i)$.
\item Let~$O_\ell(i)$ denote the set of all orientations
of the edges of~$G_\ell(i)$ which are Eulerian at every
vertex $v\in V_i\setminus X_i$.
\item Let~$O_u(i)$ denote the set of all orientations
of the edges of~$G_u(i)$ which are Eulerian at every
vertex $v\in V\setminus V_i$. 
\end{itemize}
\end{definition}

\begin{observation}
Let $G= (V,E)$ be a given graph with tree decomposition 
$(\{X_i \mid i\in I\}, T=(I,F))$, and consider any $i\in I$.
Suppose ${\cal E}\in EO(G)$.  Then ${\cal E}(i)\in O(i)$,
${\cal E}_\ell(i)\in O_\ell(i)$ and ${\cal E}_u(i)\in O_u(i)$.
\end{observation}

A neat way to group the various orientations of $O(i)$, or
$O_\ell(i)$, or $O_u(i)$ is in terms of the ``charge'' 
(outdegree - indegree) vector~$c$ induced on the~$X_i$ 
vertices by the orientation.  For a specific orientation 
${\cal O}$, we write $c=c_{{\cal O}}$ for the vector induced
by~${\cal O}$.

\begin{definition}\label{def:charge}
Let $G= (V,E)$ be a given graph with tree decomposition 
$(\{X_i \mid i\in I\}, T=(I,F))$, and let $i\in I$.  We define 
the following sets of ``charge vectors'':
\begin{itemize}
\item $C(i)\subseteq {\mathbb Z}^{|X_i|}$ is the set of all
  vectors~$c$ which can be generated by some orientation 
  on~$G(i)$.   
\item $C_\ell(i)\subseteq {\mathbb Z}^{|X_i|}$ is the set of all
  vectors~$c$ which can be generated by some orientation 
  of~$O_\ell(i)$.
\end{itemize}
\end{definition}

We now have the following observation about ``charge 
vectors''. 
\begin{observation}\label{obs:vectors1}
Suppose $G=(V,E)$ is a graph (or multi-graph) with tree 
decomposition $(\{X_i \mid i\in I\}, T=(I,F))$, and let $i\in I$. 
Then we have the following:
\begin{itemize}
\item Every vector~$c\in C(i)$ satisfies $c(x)\in \{-d_{G(i)}(x),
-d_{G(i)}(x)+2, \ldots, d_{G(i)}(x)-2, d_{G(i)}(x)\}$. 
\item Every~$c\in C_\ell(i)$ satisfies $c(x)\in \{-d_{G_\ell(i)}(x),
-d_{G_\ell(i)}(x)+2, \ldots, d_{G_\ell(i)}(x)-2, d_{G_\ell(i)}(x)\}$.  
\end{itemize}
\end{observation}

Combining Observation~\ref{obs:vectors1} together with bounded
treewidth, we can derive specific bounds on the size of $|C(i)|$ 
and $|C_\ell(i)|$.
\begin{observation}\label{obs:vectors2}
Suppose $G=(V,E)$ is a {\em simple} graph with tree 
decomposition $(\{X_i \mid i\in I\}, T=(I,F))$ of treewidth~$k$, 
and let $i\in I$. Then the following hold:
\begin{itemize}
\item By assumption of treewidth~$k$ and simplicity of the 
graph, we know $d_{G(i)}(x)\leq (k-1)$ for all $x\in X_i$.  
Therefore by Observation~\ref{obs:vectors1}, $|C(i)|\leq k^k$.
\item By simplicity, we know $d_{G_\ell(i)}(x)\leq (n-1)$ 
for all $x\in X_i$.  Therefore by Observation~\ref{obs:vectors1},   
$|C_\ell(i)|\leq n^k$.
\end{itemize}
\end{observation}

For the case of multi-graphs, we have a lesser observation:
\begin{observation}
Suppose $G=(V,E)$ is a multi-graph with tree decomposition 
$(\{X_i \mid i\in I\}, T=(I,F))$ of treewidth~$k$, and let 
$i\in I$. Let $m=|E|$. Then we have $d_{G(i)}(x) \leq (m-1)$
and $d_{G_\ell(i)}(x) \leq (m-1)$ for every $x\in X_i$.  
Hence $|C(i)|\leq m^k$ and $|C_\ell(i)|\leq m^k$.
\end{observation}

Finally we introduce one more concept which is slightly more
detailed than the concept of a ``charge vector''.  In contrast 
to charge vectors, we only define ``orientation vectors'' wrt 
the graph~$G(i)$ (which always has~$\leq k$ vertices).

\begin{definition}\label{def:orientation}
Let $G=(V,E)$ be a given graph with tree decomposition 
$(\{X_i \mid i\in I\}, T=(I,F))$, and let $i\in I$. Let
${\cal O}\in O(i)$ be any orientation on~$G(i)$.  Then we
define the {\em orientation vector} $\gamma=\gamma({\cal O})$ 
of~${\cal O}$ to be the vector~$\gamma\in {\mathbb Z}^{|X_i|*(|X_i| -1)/2}$
indexed by pairs of the form $\{x,x'\}$ for $x, x'\in X_i$, 
$x\neq x'$, where $\gamma_{\{x,x'\}}$ is set to be the number of 
edges $e=(x,x')$ which are oriented away from the vertex of 
lower index (whether~$x$ or~$x'$) in~${\cal O}$.     
\end{definition}

We now define the set of possible orientation vectors (wrt
a particular $i\in I$). 
\begin{definition}\label{def:allorients}
Let $G=(V,E)$ be a given graph with tree tree decomposition 
$(\{X_i \mid i\in I\}, T=(I,F))$, and let $i\in I$.  Then we
define the set $\Gamma(i)\subseteq {\mathbb N}_0^{|X_i|*|X_i-1|/2}$ 
to be the set of all orientation vectors~$\gamma$ which are induced 
by {\em some} orientation ${\cal O}\in O(i)$.
\end{definition}

We make a very simple observation which will be important for
the dynamic programming algorithm in~\ref{ssec:EOalg}.  Any
orientation vector~$\gamma \in \Gamma(i)$ maps to a single charge
vector $c\in C(i)$.  Therefore, given the set~$\Gamma(i)$ of orientation 
vectors for a node~$i$ of the tree decomposition, we can easily 
generate the set~$C(i)$ of charge vectors for~$i$.  

\subsection{Our algorithm}\label{ssec:EOalg}

In the light of our observations above, note that when we 
consider the root~$r$ of the tree decomposition, then $G_u(r)$
is an empty graph.  Therefore we can express the total 
number of EOs of~$G$ as follows:
\begin{eqnarray*}
|EO(G)| & = & \sum_{{\cal O}\in O(r)}\sum_{{\cal O}\in O_\ell(r)}
I_{c({\cal O})+c({\cal O}') = 0}
\end{eqnarray*}
Suppose we partition the set $O(i)$ into the disjoint sets $O(i,c)$,
for $c\in C(i)$.  The set $O(i,c)$ contains all orientations~${\cal O}$
on~$G(i)$ which have the charge vector~$c$. Clearly this partitions
$O(i)$ into disjoint sets.  Also we can partition $O_\ell(i)$ 
into the disjoint sets $O_\ell(i,c)$ for $c\in C_\ell(i)$.
Therefore we can express the number of EOs of~$G$ as follows:  
\begin{eqnarray}
|EO(G)| & = & \sum_{c\in C(r)}|O(r,c)|\times|O_\ell(r,-c)|.\label{eq:rootEOs}
\end{eqnarray}

Our algorithm will compute and maintain two sets of values, for
every node~$i\in I$ of the tree decomposition.  We define 
a second partition of~$O(i)$ with respect to orientation vectors -
for each $\gamma \in \Gamma(i)$, we let $O(i,\gamma)$ be the 
set of all orientations~${\cal O}$ on~$G(i)$ which have the 
orientation vector~$\gamma$.  The two tables we will construct,
for each~$i\in I$, are the following:
\begin{itemize}
\item[A] For every {\em orientation vector}~$\gamma\in \Gamma(i)$, we 
store the value~$o(i,\gamma) = |O(i,\gamma)|$ (the number of 
orientations of~$G(i)$ with orientation vector~$\gamma$).
\item[B] For every {\em charge vector}~$c\in C_\ell(i)$, we store 
the value~$o_\ell(i,c) = |O_\ell(i,c)|$ (the number of orientations 
of~$G_\ell(i)$ with the charge vector~$c$).
\end{itemize}
Observe that given table~A for node~$i$, we can easily compute the
set~$C(i)$, and also the value $|O(i,c)|$, for each such $c\in C(i)$.  
Therefore, when we reach the root~$r$, we can compute~$o(r,c)$ for 
all $c\in C(r)$ from Table A, and hence use those values, and table 
B, to evaluate~(\ref{eq:rootEOs}).  We have chosen to index table A. 
by orientation vectors rather than charge vectors because the extra 
information carried by the orientation vectors means that we
can exploit Table~A of the child node, when updating for the forget 
operation (which is not possible for forget, if we use charge vectors).

Note that despite the motivation given above, the advantage is not
cut-and-dried.  Re-computing the charge vectors for~$G(i)$ from 
scratch is not so computationally intensive, given that the graph~$G(i)$ 
only has a constant number of nodes. And the number of orientation 
vectors is larger than the number of charge vectors, making Table~A 
larger.  And it's messier having both concepts.  I've left it like 
this because the orientation vectors carry extra info, which we will 
probably {\em need} (and even need to extend) for the case
of Euler Tours.
 
Note that if~$G$ is simple, the table for the~$o(i,\gamma)$
values for node~$i$ is of size at most~$2^{k^2/2}$. Also,  
by Observations~\ref{obs:vectors1} and~\ref{obs:vectors2}, 
the table for the~$o_\ell(i,c)$ values is of size at most~$n^k$.

If~$G$ is not simple, the table for the~$o(i,\gamma)$ values
for node~$i$ is of size at most~$\prod_{x,x'\in X_i}(m_{\{x,x'\}}+1)$,
where~$m_{\{x,x'\}}$ is the number of edges between~$x$ and~$x'$.
This is at most $(\frac{m}{k}+1)^k \leq (\frac{2m}{k})^k$.

We now show how to build and maintain the tables for all nodes
of the tree decomposition $(\{X_i \mid i\in I\}, T=(I,F))$.  This
is done in a bottom-up dynamic programming fashion, with the
tables for node~$i$ only being built {\em after} the corresponding
tables for the child node (or nodes) of~$i$ have already been 
constructed.
Recall that we may assume that every node of the treewidth 
decomposition has at most two child nodes.  

\subsubsection{Leaf}

\noindent
Table A: In this case of a leaf node~$l$, $G(l)$ consists of 
a single vertex~$w$ and no edges.
There is only one orientation vector in~$\Gamma(l)$, which is
the ``empty" orientation vector~$\gamma^*$ (which has no pairs 
$\{x,x'\}$ with $x\neq x'$ to index it) of length~0.  The number 
of orientations of~$G(l)$ with orientation vector~$\gamma^*$ is~1, 
so we construct Table A to contain just one entry with value 
$o(l, \gamma^*)=1$. \bigskip

\noindent
Table B:  The graph $G_\ell(l)$ is the empty graph.  
There is also exactly one charge vector in~$C_\ell(l)$ - in each 
case this is  the vector~$c^*$ of length~1 which assigns charge-0 
to the single vertex~$w$. So Table B will contain the single
entry $o_\ell(l, c^*) =1$.   

\subsubsection{Introduce}

For the case of introduce, our current node~$i\in I$ has a 
single child~$i'$, and $X_{i} = X_{i'}\cup \{w\}$ for some
$w \not \in X_{i'}$.  Observe that by the properties of a 
treewidth decomposition, that for every $v\in V_{i'} \setminus X_{i'}$, 
there is no edge of the form $(w,v)$ in~$G$.  Therefore the adjacent 
vertices to~$w$ are all either in~$X_i$ or in $V\setminus V_i$,  
and the set of adjacent edges to~$w$ all belong to~$E(i)$ or
to~$E^u(i)$.

Observe that in constructing Table B for node~$i$ of the
treewidth decomposition, the charge vectors in~$C_\ell(i)$ 
will have length greater-by-1 than those in~$C_\ell(i')$, to 
include the charge on~$w$.    Similarly, in constructing 
Table A for node~$i$, the observation vectors in~$\Gamma(i)$ 
will be longer than those of~$\Gamma(i')$ by some value~$\kappa$, 
where $\kappa < k$ is the number of vertices of~$X_{i'}$ which 
are adjacent to~$x$ in~$G$. \bigskip

\noindent
Table A: We now describe how we will build Table A for node~$i$.  
First observe that~${\cal O}$ is an orientation on~$G(i)$ if and only if
${\cal O}$ is the disjoint union of some orientation~${\cal O'}$ 
on~$G(i')$ and some orientation of the adjacent edges to~$w$ 
in~$G(i)$.   In terms of {\em orientation vectors}, note that
the orientation vector~$\gamma({\cal O})$ will consist 
of~$\gamma({\cal O}')$ appended by a vector~$q=q(i,{\cal O})$ of 
length~$\kappa$ which describes, for each $v\in X_i\setminus \{w\}$, 
how many of the~$\{v,w\}$ edges are oriented {\em away}
from the vertex of lower index, whether that be~$v$ or~$w$.   
For every $v\in X_{i'}$ which is adjacent to~$w$, let $m_{v,w}$ 
denote the number of edges between~$w$ and~$v$.  Then for every 
orientation~${\cal O'}\in O(i')$, and every $q \in {\mathbb N}_0^\kappa$ 
such that $q_{v,w}\leq m_{v,w}$ for all~$v\in X_{i'}$ adjacent to~$w$, 
there are~$\binom{m_{v,w}}{q_{v,w}}$ different 
orientations~${\cal O}\in O(i)$ which induce~${\cal O}'$
on~$G(i')$ and have the orientation vector $\gamma({\cal O'}).q$.  
Therefore to create Table A, we perform the following steps:
\begin{itemize}
\item Generate all of the $\prod_{v\in X_{i'}, v\in A(w)} (m_{v,w}+1)$  
possible $q\in {\mathbb N}_0^\kappa$ vectors.
\item Create Table A for node~$i$ to contain an entry for every 
orientation vector of the form $\gamma.q$, where $\gamma$ is an 
orientation vector of node~$i'$, and 
$q\in \prod_{v\in X_{i'}, v\in A(w)}\{0, \ldots, m_{v,w}\}$. Set
$$o(i, \gamma.q) = 
o(i',\gamma)\times \prod_{v\in X_{i'}, v\in A(w)}\binom{m_{v,w}}{q_{v,w}}.$$ 
\end{itemize}
\bigskip

\noindent
Table B:  This table is indexed by charge vectors $c\in C_\ell(i)$.  
The charge vectors for node~$i$ include an extra entry for the new 
vertex~$w$.  However, since~$w$ is an isolated vertex in~$G_\ell(i)$, 
the only feasible charge for~$w$
for any $c\in C_\ell(i)$ is~$0$.  Therefore the set~$C_\ell(i)$ 
ie equivalent to $\{c.0: c\in C_\ell(i')\}$.  To create Table~B
for node~$i$, we simply take each vector~$c\in C_\ell(i')$ of
Table~B for node~$i'$, add the vector~$c.0$ to the new Table~B,
and assign~$o(i,c.0) = o(i',c)$.  Then we are finished.
  
\subsubsection{Forget}

For the case of forget, our current node~$i\in I$ has a 
single child~$i'$, and $X_{i} = X_{i'}\setminus \{w\}$ for some
$w \not \in X_{i'}$.  Let~$\kappa'$ be the number of vertices 
of~$X_i$ which are adjacent to~$w$. Also, for each $v\in X_i$
which is adjacent to~$w$, let~$m_{v,w}$ be the number of edges
between~$v$ and~$w$.  Note that in constructing Table~B for 
node~$i$, the length of charge vectors is~1 less than for 
node~$i'$, as vertex~$w$ does not belong to~$X_i$.  For Table~A, 
the orientation vectors for node~$i$ will be shorter than those 
for node~$i'$ by length~$\kappa'$.\bigskip

\noindent
Table~A:  We will use Table~A for node~$i'$ to build the new table
for its parent.   Let~$\gamma\in \Gamma(i')$ be any orientation 
vector for~$i'$. Then~$\gamma$ 
induces an orientation vector~$\widehat{\gamma}\in \Gamma(i)$
for~$i$, where~$\widehat\gamma$ is simply the vector obtained by 
deleting the orientation value for all pairs~$w,v$, 
for~$v\in X_i$.  Moreover, the number of orientations in~$O(i)$
which have the orientation vector~$\widehat\gamma$ is equal to 
\begin{eqnarray}
 \frac{o(i',\gamma)}
{\prod_{v\in X_i, v\in A(w)}\binom{m_{v,w}}{\gamma_{v,w}}}.\label{eq:restrict}
\end{eqnarray}
Therefore to construct Table A, we just scan the existing 
Table~A for node~$i'$, taking each~$\gamma$ vector in turn.  
Whenever we find a~$\gamma$ for which the restriction~$\widehat{\gamma}$ 
to~$\{u,v\} \subseteq X_i, u\neq v$, is a new orientation 
vector for~$i$, then we add a new entry to~$i$'s Table A for
$\widehat{\gamma}$, and allocate it the value 
of~(\ref{eq:restrict}).\bigskip

\noindent
Table~B: To construct Table~B, we consider the orientations~${\cal O}
\in O_\ell(i)$ on~$G_\ell(i)$.   To compare with orientations 
on~$G_\ell(i')$, recall that $E_\ell(i') \subseteq E_\ell(i)$ 
and that~$E_\ell(i)\setminus E_\ell(i')$ is the set containing 
all the edges between~$w$ and any~$v\in X_i$.  Therefore any 
orientation~${\cal O}\in O_\ell(i)$
can be partitioned into an orientation~${\cal O}'\in O_\ell(i')$ 
and some orientation of the edges between~$w$ and the vertices~$v\in X_i$.
The orientation of the edges joining~$w$ to the~$X_i$ vertices 
induces some vector~$q$ of length~$\kappa'$ where the entry~$q_{v,w}$ 
denotes the number of~$(v,w)$ edges which are directed away from 
the lower index.  The charge vector $c=c({\cal O})$, which is indexed 
by~$v\in X_i$, then satisfies the following equation, in
relation to~$c'=c({\cal O}')$ and~$q$: 
\begin{eqnarray}
c_v & = & c_v' \pm (m_{v,w}-2q_{v,w}),\label{eq:forgetcharge}
\end{eqnarray} 
where we add~$m_{v,w} -2q_{v,w}$ if~$w$ is of lower index than~$v$,
and subtract otherwise.  Note that orientations of~$O_\ell(i')$ 
with different charge vectors~$c'$ (ie from different~$O_\ell(i',c')$
sets) may contribute to orientations of~$O_\ell(i)$ with the same 
charge vector.  Equation~(\ref{eq:forgetcharge}) gives us a method to 
update Table~B for~$i$.  We process Table~B of node~$i'$, one charge 
vector at a time.  For every charge vector~$c'$ for~$i'$ (with
associated value $o(i',c')$), 
we modify~$c'$ by first deleting the index for~$w$, to give a 
vector~$c''$ of length~$|X_i|\leq k-1$.   Then we consider every 
possible vector~$q\in {\mathbb N}_0^{\kappa}$ (indexed by pairs~$v,w$ 
for~$v\in X_i$ adjacent to~$w$) such that $q_{v,w}\leq m_{v,w}$
for all $v\in X_i, v\in A(w)$.  We define $c = c(c', q)$ to be the 
vector of length~$X_i$ defined by equation~(\ref{eq:forgetcharge}).  
If~$c$ does not yet appear in Table~B for node~$i$, we create a new
entry in the table indexed by~$c$, and initialise it to~0.   Then
we compute the value 
$$\left[\prod_{v\in X_i, v\in A(w)}\binom{m_{v,w}}{q_{v,w}}\right]
\times o(i',c'),$$
and add it to the current value stored for~$c$ in the table for 
node~$i$.  By the end of the computation, we have
considered all possible $(c',q)$ pairs for $c'\in C_\ell(i')$ 
and~$q$ describing orientations of the adjacent edges to~$w$,
and we have constructed Table~B for~$i'$.

\subsubsection{Join}

In the case of a join node~$j$, we know that~$j$ has two child 
nodes~$i$ and~$i'$, and that $X_j = X_{i} = X_{i'}$.  By our
assumptions of bottom-up dynamic programming, we know the
values~$o(i,\gamma)$ and~$o(i,c)$ for all $\gamma \in \Gamma(i)$
and $c\in C_\ell(i)$ respectively (Tables A and B for node~$i$).  
Also, we know the
values~$o(i',\gamma')$ and~$o(i',c')$ for all $\gamma' \in \Gamma(i)$
and $c'\in C_\ell(i')$ respectively (Tables A and B for node~$i'$).

\bigskip

\noindent
Table A: Observe that $G(j) = G(i) = G(i')$.  Therefore, $\Gamma(j) 
= \Gamma(i)= \Gamma(i')$ and moreover, $o(j,\gamma)= o(i,\gamma)=
o(i',\gamma)$ for every $\gamma\in \Gamma(j)$.
Therefore Table~A for node~$j$ (containing all $o(j,\gamma)$ values) 
can be directly copied from the same table for node~$i$ (or for
node~$i'$) with no extra work.
\bigskip

\noindent
Table B:
Next we consider $G_\ell(j)$, and the $o_\ell(j,c)$ values for
$c\in C_\ell(j)$.  Observe that $V_j = V_i \cup V_{i'}$.  Also
note that $V_j \setminus X_j = (V_i \setminus X_i) \cup 
(V_{i'} \setminus X_{i'})$.  Also, by the rules of a treewidth 
decomposition, we know that $(V_i \setminus X_i)$ and 
$(V_{i'}\setminus X_{i'})$ are disjoint subsets of~$V$, and 
that~$G$ does not contain any edges connecting vertices of 
$V_i\setminus X_i$ with vertices of $V_{i'}\setminus X_{i'}$.
Therefore $G_\ell(i)$ and $G_\ell(i')$ are edge-disjoint 
graphs which only share the vertices~$X_i=X_{i'}=X_j$, and the 
graph~$G_\ell(j)$ is the disjoint union of $G_\ell(i)$ and
$G_\ell(i')$.  Therefore, every orientation ${\cal O}\in O_\ell(j)$
can be expressed as the composition of its induced orientation
${\cal O}_i$ on $G_\ell(i)$ and its induced orientation 
${\cal O}_{i'}$ on $G_\ell(i')$.  Note that the ``charge vector''
sets $C_\ell(j), C_\ell(i), C_\ell(i')$ are not necessarily the 
same (as the~$X_i$ vertices will have different degrees in 
the different underlying graphs); however, the vectors in the
three sets all have the same length~$|X_j|$, and (we assume) the same 
indexing on the vertices in~$X_j$.  The following equality holds 
for every $c\in C_\ell(j)$:
\begin{eqnarray}
o_\ell(j,c) & = & \sum_{c'\in C_\ell(i')} 
o_\ell(i',c')*o_\ell(i, c-c')\label{eq:dpJoin}
\end{eqnarray}
With the equation~(\ref{eq:dpJoin}), we have a simple algorithm
to compute Table~B for node~$j$ of the tree decomposition.
We simply perform a ``nested-loops'' computation on Table~B for 
node~$i$ wrt Table~B for node~$i'$.  The outside loop iterates
through all stored~$(i,c)$ entries for node~$i$; the inside 
loop iterates through all stored~$(i', c')$ entries for node~$i'$.
At each step, we compute $c+c'$.  If the entry~$(j,c+c')$ in~$j$'s 
partial (b)-table has already been instantiated, then  we add the
value~$o_\ell(i,c)\times o_\ell(i',c')$ to the current value stored 
there.  If there is no value yet stored for~$(j,c+c')$, then we 
initialise the value to $o_\ell(i,c)\times o_\ell(i',c')$.
By equation~(\ref{eq:dpJoin}), we are guaranteed that when the
nested-loops computation is finished, the (b)-table for node~$j$
now contains all the correct $o_\ell(j,c)$ values for 
$c\in C_\ell(j)$.  

Observe that for the join computation, the issue of whether~$G$
was a simple graph or a multi-graph is not relevant (except in 
bounding the size of the tables).

\section{Euler tours (to be completed)}

We are given an Eulerian graph $G=(V,E)$ with treewidth~$k$
for some constant $k\in {\mathbb N}_0$.  We also have an 
explicit nice tree decomposition for $G$.   

The standard representation of an Euler tour is as a 
sequence
$$v_{\tau(1)}, v_{\tau(2)}, \ldots, v_{\tau(m)}$$
of vertices of the graph $G$, such that for every 
edge $e=(u,v)\in E$, there is exactly one $ 1\leq i\leq m$ 
such that $\{ v_{\tau(i)}, v_{\tau((i \bmod m)+1)}\} =
\{u,v\}$.  Assuming that we fix the initial edge 
$(v_{\tau(1)},v_{\tau(2)})= (v_1,v_2)$ of every tour for some 
edge~$e=(v_1,v_2)$ of the graph, there is a unique such 
representation of every element of $ET(G)$.

In designing the algorithm for the bounded treewidth case, 
we will consider an alternative representation in terms of 
{\em Orbs}.  Observe that every Euler tour~$T\in ET(G)$ (taken 
with the fixed starting edge~$e_1$) induces a unique Eulerian 
Orientation ${\cE}_T$ on the graph.  However, there may be many 
Euler tours with this same Eulerian Orientation.   To distinguish
between them, define the {\em last exit arc} for a vertex~$v\in V$ 
to be the arc~$(v,v')$ formed by the last occurrence of~$v$ in~$T$
with its successor vertex~$v'$.   It is well-known that the collection
of last exit arcs for all $v\in V\setminus \{v_1\}$ forms an 
unique {\em Arborescence} $\cA_T$ directed inwards 
to~$v_1$.  Moreover it can be shown that any Arborescence~$\cA$
of an Eulerian Orientation~$\cE$ of~$G$ corresponds to exactly
$$\prod_{v\in V} (\frac{d(v)}{2}-1)!$$ different Euler tours with the
orientation~$\cE$.  We define an {\em Orb} $\cA$ of $G$ to be any 
pair $(\cE,\cA)$ such that $\cE\in EO(G)$ and~$\cA$ is an in-Arborescence 
of~$\cE$.  In counting Euler Tours, it is sufficient to instead 
count Orbs, and take that count times the factor above as our
Euler tour count.
}

\end{document}